\documentclass[a4paper, USenglish, numberwithinsect, cleveref, autoref, thm-restate]{lipics-v2019}
\usepackage{xparse}
\hypersetup{hypertexnames=false} 

\newcommand{\otyp}{\mathsf{o}}
\newcommand{\arr}{\mathbin{\to}}
\newcommand{\Arr}{\mathbin{\Rightarrow}}
\newcommand{\set}[1]{\{#1\}}
\newcommand{\dom}{\mathrm{dom}}
\newcommand{\setof}[2]{\{#1\mid#2\}}
\newcommand{\Nat}{\mathbb{N}}
\newcommand{\Gg}{\mathcal{G}}
\newcommand{\Ff}{\mathcal{F}}

\newcommand{\Oo}{\mathcal{O}}
\newcommand{\Nn}{\mathcal{N}}
\newcommand{\Rr}{\mathcal{R}}

\newcommand{\Xx}{\mathcal{X}}
\newcommand{\Yy}{\mathcal{Y}}
\newcommand{\Zz}{\mathcal{Z}}
\newcommand{\varT}{\mathsf{T}}
\newcommand{\varX}{\mathsf{X}}
\newcommand{\varY}{\mathsf{Y}}
\newcommand{\varZ}{\mathsf{Z}}
\newcommand{\varx}{\mathsf{x}}
\newcommand{\vary}{\mathsf{y}}
\newcommand{\ord}{\mathsf{ord}}
\newcommand{\gar}{\mathsf{gar}}
\newcommand{\tp}{\mathsf{tp}}
\newcommand{\tr}{\mathsf{tr}}
\newcommand{\rew}{\longrightarrow}
\newcommand{\erew}{\rightsquigarrow}
\newcommand{\esubst}[3]{#1\langle#2/#3\rangle}
\newcommand{\esubstdots}[5]{#1\langle#2/#3\rangle\dots\langle#4/#5\rangle}
\NewDocumentCommand{\symb}{ o o o o o o }{\langle\IfValueT{#1}{#1}\IfValueT{#2}{,#2}\IfValueT{#3}{,#3}\IfValueT{#4}{,#4}\IfValueT{#5}{,#5}\IfValueT{#6}{,#6}\rangle}
\Crefname{lemma}{Lemma}{Lemmata}
\Crefname{equation}{Equality}{Equalities}
\Crefname{equalities}{Equalities}{Equalities}\creflabelformat{equalities}{(#2#1#3)}

\title{Higher-Order Nonemptiness Step by Step}

\author{Paweł Parys}{Institute of Informatics, University of Warsaw, Poland}{parys@mimuw.edu.pl}{https://orcid.org/0000-0001-7247-1408}{}

\authorrunning{P. Parys}

\Copyright{Paweł Parys}

\ccsdesc[500]{Theory of computation~Rewrite systems}

\keywords{Higher-order grammars, Nonemptiness, Model-checking, Transformation, Order reduction}

\supplement{Coq formalisation: \href{https://github.com/pparys/ho-transform-sbs}{https://github.com/pparys/ho-transform-sbs}}

\hideLIPIcs
\nolinenumbers

\begin{document}

\maketitle

\begin{abstract}
	We show a new simple algorithm that checks whether a given higher-order grammar generates a nonempty language of trees.
	The algorithm amounts to a procedure that transforms a grammar of order $n$ to a grammar of order $n-1$, preserving nonemptiness, and increasing the size only exponentially.
	After repeating the procedure $n$ times, we obtain a grammar of order $0$, whose nonemptiness can be easily checked.
	Since the size grows exponentially at each step, the overall complexity is $n$\textsf{-EXPTIME}, which is known to be optimal.
	More precisely, the transformation (and hence the whole algorithm) is linear in the size of the grammar, assuming that the arity of employed nonterminals is bounded by a constant.
	The same algorithm allows to check whether an infinite tree generated by a higher-order recursion scheme is accepted by an alternating safety (or reachability) automaton,
	because this question can be reduced to the nonemptiness problem by taking a product of the recursion scheme with the automaton.
	
	A proof of correctness of the algorithm is formalised in the proof assistant Coq.
	Our transformation is motivated by a similar transformation of Asada and Kobayashi (2020) changing a word grammar of order $n$ to a tree grammar of order $n-1$.
	The step-by-step approach can be opposed to previous algorithms solving the nonemptiness problem ``in one step'', being compulsorily more complicated.
\end{abstract}

\section{Introduction}

	Higher-order grammars, also known as higher-order OI grammars~\cite{Damm82,IO-OI}, generalize context-free grammars:
	nonterminals of higher-order grammars are allowed to take arguments.
	Such grammars have been studied actively in recent years, in the context of automated verification of higher-order programs.
	In this paper we concentrate on a very basic problem of language nonemptiness: is the language generated by a given higher-order grammar nonempty.
	This problem, being easy for most devices, is not so easy for higher-order grammars.
	Indeed, it is $n$\textsf{-EXPTIME}-complete for grammars of order $n$~\cite{ho-complexity}.
	
	We give a new simple algorithm solving the language nonemptiness problem.
	The algorithm amounts to a procedure that transforms a grammar of order $n$ to a grammar of order $n-1$, preserving nonemptiness, and increasing the size only exponentially.
	After repeating the procedure $n$ times, we obtain a grammar of order $0$, whose nonemptiness can be easily checked.
	Since the size grows exponentially at each step, we reach the optimal overall complexity of $n$\textsf{-EXPTIME}.
	In a more detailed view, the complexity looks even better:
	the size growth is exponential only in the arity of types appearing in the grammar;
	if the maximal arity is bounded by a constant, the transformation (and hence the whole algorithm) is linear in the size of the grammar.
	
	While a higher-order grammar is a generator of a language of (finite) trees, virtually the same object can be seen as a generator of a single infinite tree (encompassing the whole language).
	In this context, the grammars are called higher-order recursion schemes.
	The nonemptiness problem for grammars is easily equivalent to the question whether the tree generated by a given recursion scheme is accepted by a given alternating safety (or reachability) automaton;
	for the right-to-left reduction, it is enough to product the recursion scheme with the automaton.
	Thus, our algorithm solves also the latter problem, called a model-checking problem.
	This problem is decidable and $n$\textsf{-EXPTIME}-complete not only for safety or reachability automata, but actually for all parity automata,
	with multiple proofs using game semantics~\cite{Ong-lics}, collapsible pushdown automata~\cite{collapsible}, intersection types~\cite{Kobayashi-Ong}, or Krivine machines~\cite{Krivine-machine},
	and with several extensions~\cite{global,reflection,effective-selection,model,wmsou}.
	The problem for safety automata was tackled in particular by Aehlig~\cite{Aehlig07} and by Kobayashi~\cite{Kobayashi-POPL09}.
	To those algorithms we add another one.
	The main difference between our algorithm and all the others is that we solve the problem step by step, repeatedly reducing the order by one,
	while most previous algorithms work ``in one step'', being compulsorily more complicated.
	The only proofs that have been reducing the order by one, were proofs using collapsible pushdown automata~\cite{collapsible,reflection,effective-selection},
	being very technical (and contained only in unpublished appendices).
	A reduction of order was also possible for a subclass of recursion schemes, called \emph{safe} recursion schemes~\cite{easy-trees},
	but it was not known how to extend it to all recursion schemes.
	
	Comparing the two variants of the model-checking problem for higher-order recursion schemes%
	---involving safety and reachability automata, and involving all parity automata---we have to mention two things.
	First, while most theoretical results can handle all parity automata,
	actual tools solving this problem in practice mostly deal only with safety and reachability automata
	(called also trivial and co-trivial automata)~\cite{Kobayashi-jacm,horsat,ZDD,preface}.
	Second, there exists a polynomial-time (although nontrivial) reduction from the variant involving parity automata to the variant involving safety automata~\cite{parity2safety}.
	
	Our transformation is directly motivated by a recent paper of Asada and Kobayashi~\cite{word2tree}.
	They show how to transform a grammar of order $n$ generating a language of words to a grammar of order $n-1$ generating a language of trees,
	so that words of the original language are written in leaves of trees of the new language.
	Unexpectedly, this transformation increases the size of the grammar only polynomially.
	Our transformation is quite similar, but we start from a grammar generating a language of trees, not words.
	In effect, on the one hand, we do not say anything specific about the language after the transformation (except that nonemptiness is preserved),
	and on the other hand, the size growth is exponential, not polynomial.

\section{Preliminaries}

	For a number $k\in\Nat$ we write $[k]$ for $\set{1,\dots,k}$.

	The set of \emph{(simple) types} is constructed from a unique ground type $\otyp$ using a binary operation $\arr$; 
	namely $\otyp$ is a type, and if $\alpha$ and $\beta$ are types, so is $\alpha\arr\beta$.
	By convention, $\arr$ associates to the right, that is, $\alpha\arr\beta\arr\gamma$ is understood as $\alpha\arr(\beta\arr\gamma)$.
	We often abbreviate $\underbrace{\alpha\arr\dots\arr\alpha}_\ell\to\beta$ as $\alpha^\ell\arr\beta$.
	The \emph{order} of a type $\alpha$, denoted $\ord(\alpha)$, is defined by induction: 
	$\ord(\alpha_1\arr\dots\arr\alpha_k\arr\otyp)=\max(\set{0}\cup\setof{\ord(\alpha_i)+1}{i\in[k]})$;
	for example $\ord(\otyp)=0$, $\ord(\otyp\arr\otyp\arr\otyp)=1$, and $\ord((\otyp\arr\otyp)\arr\otyp)=2$.

	Having a finite set of typed nonterminals $\Xx$, and a finite set of typed variables $\Yy$, 
	\emph{terms} over $(\Xx,\Yy)$ are defined by induction:
	\begin{itemize}
	\item	every nonterminal $X\in\Xx$ of type $\alpha$ is a term of type $\alpha$;
	\item	every variable $y\in\Yy$ of type $\alpha$ is a term of type $\alpha$;
	\item	if $K_1,\dots,K_k$ are terms of type $\otyp$, then $\bullet\symb[K_1][\dots][K_k]$ and $\oplus\symb[K_1][\dots][K_k]$ are terms of type $\otyp$;
	\item	if $K$ is a term of type $\alpha\arr\beta$, and $L$ is a term of type $\alpha$, then $K\,L$ is a term of type $\beta$.
	\end{itemize}
	The type of a term $K$ is denoted $\tp(K)$.
	The order of a term $K$, written $\ord(K)$, is defined as the order of its type.
	We write $\Omega$ for $\oplus\symb{}$, and $\bullet$ for $\bullet\symb{}$.

	The construction $\oplus\symb[K_1][\dots][K_k]$ is an alternative; such a term reduces to one of the terms $K_1,\dots,K_k$.
	This construction is used to introduce nondeterminism to grammars (defined below).
	In the special case of $k=0$ (when we write $\Omega$) no reduction is possible; thus $\Omega$ denotes divergence.
	
	The construction $\bullet\symb[K_1][\dots][K_k]$ can be seen as a generator of a tree node with $k$ children;
	subtrees starting in these children are described by the terms $K_1,\dots,K_k$.
	In a usual presentation, nodes are labeled by letters from some finite alphabet.
	In this paper, however, we do not care about the exact letters contained in generated trees, only about language nonemptiness, 
	hence we do not write these letters at all (in other words, we use a single-letter alphabet, where $\bullet$ is the only letter).
	Actually, in the sequel we even do not consider trees; we rather say that $\bullet\symb[K_1][\dots][K_k]$ is convergent if all $K_1,\dots,K_k$ are convergent
	(which can be rephrased as: the language generated from $\bullet\symb[K_1][\dots][K_k]$ is nonempty if the languages generated from all $K_1,\dots,K_k$ are nonempty).
	
	A \emph{(higher-order) grammar} is a tuple $\Gg=(\Xx,X_0,\Rr)$, where $\Xx$ a finite set of typed nonterminals,
	$X_0\in\Xx$ is a starting nonterminal of type $\otyp$,
	and $\Rr$ a function assigning to every nonterminal $X\in\Xx$ a rule of the form $X\,y_1\,\dots\,y_k\to R$,
	where $\tp(X)=(\tp(y_1)\arr\dots\arr\tp(y_k)\arr\otyp)$, and $R$ is a term of type $\otyp$ over $(\Xx,\set{y_1,\dots,y_k})$.
	The order of a grammar is defined as the maximum of orders of its nonterminals.
	
	Having a grammar $\Gg=(\Xx,X_0,\Rr)$, 
	for every set of variables $\Yy$ we define a \emph{reduction relation} $\rew_\Gg$ between terms over $(\Xx,\Yy)$ and sets of such terms, as the least relation such that
	\begin{bracketenumerate}
	\item	$X\,K_1\,\dots\,K_k\rew_\Gg \set{R[K_1/y_1,\dots,K_k/y_k]}$ if the rule for $X$ is $X\,y_1\,\dots\,y_k\to R$, 
		where $R[K_1/y_1,\allowbreak\dots,K_k/y_k]$ denotes the term obtained from $R$ by substituting $K_i$ for $y_i$ for all $i\in[k]$,
	\item	$\bullet\symb[K_1][\dots][K_k]\rew_\Gg\set{K_1,\dots,K_k}$, and
	\item	$\oplus\symb[K_1][\dots][K_k]\rew_\Gg\set{K_i}$ for every $i\in[k]$.
	\end{bracketenumerate}
	
	We say that a term $M$ is \emph{$\Gg$-convergent} if $M\rew_\Gg\Nn$ for some set $\Nn$ of $\Gg$-convergent terms.
	This is an inductive definition; in particular, the base case is when $M\rew_\Gg\emptyset$.
	In other words, M is \emph{$\Gg$-convergent} if there is a finite tree labeled by terms where for each node, the node and its children satisfy one of (1)-(3).
	Moreover, the grammar $\Gg$ is \emph{convergent} if its starting nonterminal $X_0$ is $\Gg$-convergent.

\section{Transformation}\label{sec:transformation}

	In this section we present a transformation, called \emph{order-reducing transformation}, resulting in the main theorem of this paper:
	
	\begin{theorem}\label{thm:main}
		For any $n\geq 1$, there exists a transformation from order-$n$ grammars to order-$(n-1)$ grammars,
		and a polynomial $p_n$ such that, for any order-$n$ grammar $\Gg$, the resulting grammar $\Gg^\dag$ is convergent if and only if $\Gg$ is convergent, and $|\Gg^\dag|\leq 2^{p_n(|\Gg|)}$.
	\end{theorem}

\subparagraph{Intuitions.}

	Let us first present intuitions behind our transformation.
	While reducing the order, we have to replace, in particular, order-$1$ functions by order-$0$ terms.
	Consider for example a term $K\,L$ of type $\otyp$, where $K$ has type $\otyp\arr\otyp$.
	Notice that $L$ generates trees that are inserted somewhere in contexts generated by $K$.
	Thus, when is $K\,L$ convergent?
	There are two possibilities.
	First, maybe $K$ is convergent without using its argument at all.
	Second, maybe $K$ can be convergent but only using its argument, and then $L$ also has to be convergent.
	Notice that in the first case $K\,\Omega$ is convergent (i.e., $K$ is convergent even if the argument is not convergent),
	and in the second case $K\,\bullet$ is convergent (i.e., $K$ is convergent if its argument is convergent).
	In the transformation, we transform $K$ into two order-$0$ terms, $K_0$ and $K_1$ corresponding to $K\,\Omega$ and $K\,\bullet$,
	and then we replace $K\,L$ by $\oplus\symb[K_0][\bullet\symb[K_1][L]]$.

	As a full example, consider an order-$1$ grammar with the following rules:
	\begin{align*}
		\varX\to \varY\,\varZ,&&
		\varY\,\varx\to\oplus\symb[\bullet][\varx],&&
		\varZ\to\bullet.
	\end{align*}
	It will be transformed to the order-$0$ grammar with the following rules:
	\begin{align*}
		\varX\to\oplus\symb[\varY_0][\bullet\symb[\varY_1][\varZ]],&&
		\varY_0\to\oplus\symb[\bullet][\Omega],&&
		\varY_1\to\oplus\symb[\bullet][\bullet],&&
		\varZ\to\bullet.
	\end{align*}
	Notice that the original grammar is convergent ``for two reasons'': the $\oplus$ node in the rule for $\varY$ may reduce either to the first possibility (i.e., to $\bullet$),
	or to the second possibility (i.e., to $\varx$), in which case convergence follows from convergence of the argument $\varZ$.
	This is reflected by the two possibilities available for the $\oplus$ node in the new rule for $\varX$:
	we either choose the first possibility and we depend only on convergence of $\varY_0$,
	or we choose the the second possibility and we depend on convergence of both $\varY_1$ and $\varZ$.
	Notice that after replacing the (old and new) rule for $\varZ$ by $\varZ\to\Omega$, the modified grammars remain convergent thanks to the first possibility above.
	Likewise, after replacing the original rule for $\varY$ by $\varY\,\varx\to\varx$, the new rules will be $\varY_0\to\Omega$ and $\varY_1\to\bullet$,
	and the modified grammars remain convergent thanks to the second possibility above.
	However, after applying both these replacements simultaneously, the grammars stop to be convergent.
	
	If our term $K$ takes multiple order-$0$ arguments, say we have $K\,L_1\,\dots\,L_k$,
	while transforming $K$ we need $2^k$ variants of the term: each of the arguments may be either used (replaced by $\bullet$) or not used (replaced by $\Omega$).
	This is why we have the exponential blow-up.
	Let us compare this quickly with the transformation of Asada and Kobayashi~\cite{word2tree}, which worked for grammars generating words
	(i.e., trees where every node has at most one child).
	In their case, at most one of the arguments $L_i$ could be used, so they needed only $k+1$ variants of $K$; this is why their transformation was polynomial.
	
	For higher-order grammars we apply the same idea: functions of order $1$ are replaced by terms of order $0$,
	and then the order of any higher-order function drops down by one.
	For example, consider a grammar with the following rules:
	\begin{align*}
		\varX\to\varT\,\varY,&&
		\varT\,\vary\to\vary\,(\vary\,\bullet),&&
		\varY\,\varx\to\oplus\symb[\bullet][\varx].
	\end{align*}
	The nonterminal $\varY$ is again of type $\otyp\arr\otyp$, hence it is replaced by two nonterminals $\varY_0,\varY_1$ of type $\otyp$,
	describing the situation when the parameter $\varx$ is either not used or used.
	Likewise, the corresponding parameter $\vary$ of $\varT$ is replaced by two parameters $\vary_0,\vary_1$.
	The resulting grammar will have the following rules:
	\begin{align*}
		\varX\to\varT\,\varY_0\,\varY_1,&&
		\varT\,\vary_0\,\vary_1\to\oplus\symb[\vary_0][\bullet\symb[\vary_1][\oplus\symb[\vary_0][\bullet\symb[\vary_1][\bullet]]]],&&
		\varY_0\to\oplus\symb[\bullet][\Omega],&&
		\varY_1\to\oplus\symb[\bullet][\bullet].
	\end{align*}

\subparagraph{Formal definition.}	

	We now formalize the above intuitions.
	Having a type, we are interested in cutting off its suffix being of order $1$.
	Thus, we use the notation $\alpha_1\arr\dots\arr\alpha_k\Arr\otyp^\ell\arr\otyp$ for a type $\alpha_1\arr\dots\arr\alpha_k\arr\otyp^\ell\arr\otyp$ such that either $k=0$ or $\alpha_k\neq\otyp$.
	\pagebreak[3]
	Notice that every type $\alpha$ can be uniquely represented in this form.
	We remark that some among the types $\alpha_1,\dots,\alpha_{k-1}$ (but not $\alpha_k$) may be $\otyp$.
	For a type $\alpha$ we write $\gar(\alpha)$ (``ground arity'') for the number $\ell$ for which we can write $\alpha=(\alpha_1\arr\dots\arr\alpha_k\Arr\otyp^\ell\arr\otyp)$;
	we also extend this to terms: $\gar(M)=\gar(\tp(M))$.
	
	We transform terms of type $\alpha$ to terms of type $\alpha^\dag$, which is defined by induction:
	\begin{align*}
		(\alpha_1\arr\dots\arr\alpha_k\Arr\otyp^\ell\arr\otyp)^\dag = \left((\alpha_1^\dag)^{2^{\gar(\alpha_1)}}\arr\dots\arr(\alpha_k^\dag)^{2^{\gar(\alpha_k)}}\arr\otyp\right).
	\end{align*}
	Thus, we remove all trailing order-$0$ arguments, and we multiplicate (and recursively transform) remaining arguments.

	For a finite set $S$, we write $2^S$ for the set of functions $A\colon S\to\set{0,1}$.
	Moreover, we assume some fixed order on functions in $2^S$, 
	and we write $P\,(Q_A)_{A\in 2^S}$ for an application $P\,Q_{A_1}\,\dots\,Q_{A_{2^{|S|}}}$, where $A_1,\dots,A_{2^{|S|}}$ are all the function from $2^S$ listed in the fixed order.
	The only function in $2^\emptyset$ is denoted $\emptyset$.
	
	Fix a grammar $\Gg=(\Xx,X_0,\Rr)$.
	For every nonterminal $X$ and for every function $A\in2^{[\gar(X)]}$ we consider a nonterminal $X_A^\dag$ of type $(\tp(X))^\dag$.
	As the new set of nonterminals we take $\Xx^\dag=\setof{X_A^\dag}{X\in\Xx,A\in2^{[\gar(X)]}}$.
	Likewise, for every variable $y$ and for every function $A\in2^{[\gar(y)]}$ we consider a variable $y_A^\dag$ of type $(\tp(y))^\dag$,
	and for a set of variables $\Yy$ we denote $\Yy^\dag=\setof{y_A^\dag}{y\in\Yy,A\in2^{[\gar(y)]}}$.

	We now define a function $\tr$ transforming terms.
	Its value $\tr(A,Z,M)$ is defined when $M$ is a term over some $(\Xx,\Yy)$, and $A\in2^{[\gar(M)]}$,
	and $Z\colon\Yy\rightharpoonup\set{0,1}$ is a partial function such that $\dom(Z)$ contains only variables of type $\otyp$.
	The intention is that $A$ specifies which among trailing order-$0$ arguments can be used, and $Z$ specifies which order-$0$ variables (among those in $\dom(Z)$) can be used.
	The transformation is defined by induction on the structure of $M$, as follows:
	\begin{bracketenumerate}
	\item	$\tr(A,Z,X)=X_A$ for $X\in\Xx$;
	\item	$\tr(A,Z,y)=y_A$ for $y\in\Yy\setminus\dom(Z)$;
	\item	$\tr(A,Z,z)=\Omega$ if $Z(z)=0$;
	\item	$\tr(A,Z,z)=\bullet$ if $Z(z)=1$;
	\item	$\tr(\emptyset,Z,\bullet\symb[K_1][\dots][K_k])=\bullet\symb[\tr(\emptyset,Z,K_1)][\dots][\tr(\emptyset,Z,K_k)]$;
	\item	$\tr(\emptyset,Z,\oplus\symb[K_1][\dots][K_k])=\oplus\symb[\tr(\emptyset,Z,K_1)][\dots][\tr(\emptyset,Z,K_k)]$;
	\item	$\tr(A,Z,K\,L)=\oplus\symb[\tr(A[\ell+1\mapsto0],Z,K)][\bullet\symb[\tr(A[\ell+1\mapsto1],Z,K)][\tr(\emptyset,Z,L)]]$ if $\tp(K)=(\otyp^{\ell+1}\arr\otyp)$;
	\item	$\tr(A,Z,K\,L)=(\tr(A,Z,K))\,(\tr(B,Z,L))_{B\in2^{[\gar(L)]}}$ if $\tp(K)=(\alpha_1\arr\dots\arr\alpha_k\Arr\otyp^\ell\arr\otyp)$ with $k\geq 1$.
	\end{bracketenumerate}

	For every rule $X\,y_1\,\dots\,y_k\,z_1\,\dots\,z_\ell\to R$ in $\Rr$, where $\ell=\gar(X)$, 
	and for every function $A\in2^{[\ell]}$, 
	to $\Rr^\dag$ we take the rule
	\begin{align*}
		X_A^\dag\,(y_{1,B}^\dag)_{B\in2^{[\gar(y_1)]}}\,\dots\,(y_{k,B}^\dag)_{B\in2^{[\gar(y_k)]}}\to\tr(\emptyset,[z_i\mapsto A(\ell+1-i)\mid i\in[\ell]],R).
	\end{align*}
	In the function $A$ it is more convenient to count arguments from right to left (then we do not need to shift the domain in Case (7) above),
	but it is more natural to have variables $z_1,\dots,z_\ell$ numbered from left to right;
	this is why in the rule for $X_A^\dag$ we assign to $z_i$ the value $A(\ell+1-i)$, not $A(i)$.
	
	Finally, the resulting grammar $\Gg^\dag$ is $(\Xx^\dag,X_{0,\emptyset}^\dag,\Rr^\dag)$.

\section{Complexity}

	In this section we analyze complexity of our transformation.
	First, we formally define the \emph{size} of a grammar.
	The size of a term is defined by induction on its structure:
	\begin{gather*}
		|X|=|y|=1,\qquad
		|K\,L|=1+|K|+|L|,\\
		|{\bullet}\symb[K_1][\dots][K_k]|=|{\oplus}\symb[K_1][\dots][K_k]|=1+|K_1|+\dots+|K_k|.
	\end{gather*}
	Then $|\Gg|$, the size of $\Gg$ is defined as the sum of $|R|+k$ over all rules $X\,y_1\,\dots\,y_k\to R$ of $\Gg$.
	In Asada and Kobayashi~\cite{word2tree} such a size is called \emph{Curry-style} size; 
	it does not include sizes of types of employed variables.

	We say that a type $\alpha$ is a \emph{subtype} of a type $\beta$ if either $\alpha=\beta$,
	or $\beta=(\beta_1\arr\beta_2)$ and $\alpha$ is a subtype of $\beta_1$ or of $\beta_2$.
	We write $A_\Gg$ for the largest arity of subtypes of types of nonterminals in a grammar $\Gg$.
	Notice that types of other objects appearing in $\Gg$, namely variables and subterms of right sides of rules, are subtypes of types of nonterminals,
	hence their arity is also bounded by $A_\Gg$.
	It is reasonable to consider large grammars, consisting of many rules, where simultaneously the maximal arity $A_\Gg$ is respectively small.

	While the exponential bound mentioned in \cref{thm:main} is obtained by applying the order-reducing transformation to an arbitrary grammar,
	the complexity becomes slightly better if we first apply a preprocessing step.
	This is in particular necessary, if we want to obtain linear dependence in the size of $\Gg$ (and exponential only in the maximal arity $A_\Gg$).
	The preprocessing, making sure that the grammar is in a \emph{simple form} (defined below) amounts to splitting large rules into multiple smaller rules.
	A similar preprocessing is present already in prior work~\cite{Kobayashi-jacm,word2tree,diagonal},
	however our definition of a simple form is slightly more liberal,
	so that the order reduction applied to a grammar in a normal form gives again a grammar in a normal form.

	An \emph{application depth} of a term $R$ is defined as the maximal number of applications on a single branch in $R$,
	where a compound application $K\,L_1\,\dots\,L_k$ counts only once.
	More formally, we define by induction:
	\begin{align*}
		&\mathsf{ad}(\bullet\symb[K_1][\dots][K_k])=\mathsf{ad}(\oplus\symb[K_1][\dots][K_k])=\max\set{\mathsf{ad}(K_i)\mid i\in[k]},\\
		&\mathsf{ad}(X\,K_1\,\dots\,K_k)=\mathsf{ad}(y\,K_1\,\dots\,K_k)=\max(\set{0}\cup\set{\mathsf{ad}(K_i)+1\mid i\in[k]}).
	\end{align*}
	We say that a grammar $\Gg$ is in a \emph{simple form} if the right side of each its rule has application depth at most $2$.
	
	Any grammar $\Gg$ can be converted to a grammar in a simple form, as follows.
	Consider a rule $X\,y_1\,\dots\,y_k\to R$, and a subterm of $R$ of the form $f\,K_1\,\dots\,K_m$, where $f$ is a nonterminal or a variable,
	but some $K_i$ already has application depth $2$.
	Then we replace the occurrence of $K_i$ with $Y\,y_1\,\dots\,y_k$ (being a term of application depth $1$) for a fresh nonterminal $Y$,
	and we add the rule $Y\,y_1\,\dots\,y_k\,x_1\,\dots\,x_s\to K_i\,x_1\,\dots\,x_s$ (whose right side already has application depth $2$; 
	the additional variables $x_1,\dots,x_s$ are added to ensure that the type is $\otyp$).
	By repeating such a replacement for every ``bad'' subterm of every rule, we clearly obtain a grammar in a simple form.
	
	\begin{lemma}\label{simpl-complexity}
		Let $\Gg'$ be the grammar in a simple form obtained by the above simplification procedure from a grammar $\Gg$.
		Then $\ord(\Gg')=\ord(\Gg)$, and $A_{\Gg'}\leq 2A_\Gg$, and $|\Gg'|=\Oo(A_\Gg\cdot|\Gg|)$.
		The procedure can be performed in time linear in its output size.
	\end{lemma}
	
	\begin{proof}
		The parts about the order and about the running time are obvious.
		
		Types of nonterminals originating from $\Gg$ remain unchanged.
		The type of a fresh nonterminal $Y$ introduced in the procedure is of the form $\alpha_1\arr\dots\arr\alpha_k\arr\beta_1\arr\dots\arr\beta_s\arr\otyp$,
		where all $\alpha_i$ and $\beta_i$ are types present also in $\Gg$.
		The arity of the whole type is $k+s$, where $k$ is the arity of the original nonterminal $X$ (hence it is bounded by $A_\Gg$),
		and $s$ is bounded by the arity of the type of $K_i$ (hence also by $A_\Gg)$.
		
		In order to bound the size of the resulting grammar, 
		notice that the considered replacement is performed at most once for every subterm of the right side of every rule, hence the number of replacements is bounded by $|\Gg|$.
		Each such a replacement increases the size of the grammar by at most $\Oo(A_\Gg)$.
	\end{proof}
	
	\begin{lemma}\label{trans-complexity}
		For every grammar $\Gg$ in a simple form, the grammar $\Gg^\dag$ (i.e., the result of the order-reducing transformation)
		is also in a simple form, and $\ord(\Gg^\dag)=\max(0,\ord(\Gg)-1)$, and $A_{\Gg^\dag}\leq A_\Gg\cdot 2^{A_\Gg}$, 
		and $|\Gg^\dag|=\Oo(|\Gg|\cdot 2^{5\cdot A_\Gg})$.
		Moreover, the transformation can be performed in time linear in its output size.
	\end{lemma}
	
	\begin{proof}
		The part about the running time is obvious.
		It is also easy to see by induction that $\ord(\alpha^\dag)=\max(0,\ord(\alpha)-1)$.
		It follows that the order of the grammar satisfies the same equality,
		because nonterminals of $\Gg^\dag$ have type $\alpha^\dag$ for $\alpha$ being the type of a corresponding nonterminal of $\Gg$.
		
		Recall that in the type $\alpha^\dag$ obtained from $\alpha=(\alpha_1\arr\dots\arr\alpha_k\arr\otyp)$,
		every $\alpha_i$ either disappears or becomes (transformed and) repeated $2^{\gar(\alpha_i)}$ times, that is, at most $2^{A_\Gg}$ times.
		This implies the inequality concerning $A_{\Gg^\dag}$.
		
		Every compound application can be written as $f\,K_1\,\dots\,K_k\,L_1\,\dots\,L_\ell$, where $f$ is a nonterminal or a variable, and $\ell=\gar(f)$.
		In such a term, every $K_i$ (after transforming) becomes repeated $2^{\gar(K_i)}$ times, that is, at most $2^{A_\Gg}$ times.
		Then, for every $L_i$ we duplicate the outcome and we append a small prefix;
		this duplication happens $\ell$ times, that is, at most $A_\Gg$ times.
		In consequence, we easily see by induction that while transforming a term of application depth $d$, its size gets multiplicated by at most $O(2^{2d\cdot A_\Gg})$.
		Moreover, every nonterminal $X$ is repeated $2^{\gar(X)}$ times, that is, at most $2^{A_\Gg}$ times.
		Because the application depth of right sides of rules is at most $2$, this bounds the size of the new grammar by $\Oo(|\Gg|\cdot 2^{5\cdot A_\Gg})$.
		
		Looking again at the above description of the transformation, we can notice that the application depth cannot grow;
		in consequence the property of being in a simple form is preserved.
	\end{proof}
	
	Thus, if we want to check nonemptiness of a grammar $\Gg$ of order $n$,
	we can first convert it to a simple form, and then apply the order-reducing transformation $n$ times.
	This gives us a grammar of order $0$, whose nonemptiness can be checked in linear time.
	By \cref{simpl-complexity,trans-complexity}, the whole algorithm works in time $n$-fold exponential in $A_\Gg$ and linear in $|\Gg|$.
	
	If the original grammar $\Gg$ generates a language of words, we can start by applying the polynomial-time transformation of Asada and Kobayashi~\cite{word2tree},
	which converts $\Gg$ into an equivalent grammar of order $n-1$ (generating a language of trees); then we can continue as above.
	Because their transformation is also linear in $|\Gg|$, and increases the arity only quadratically,
	in this case we obtain an algorithm working in time $(n-1)$-fold exponential in $A_\Gg$ and linear in $|\Gg|$.

\section{Correctness}\label{sec:correctness}

	In this section we finish a proof of \cref{thm:main} by showing that the grammar $\Gg^\dag$ resulting from transforming a grammar $\Gg$ is convergent if and only if
	the original grammar $\Gg$ is convergent.
	This proof is also formalised in the proof assistant Coq, and available at GitHub (\href{https://github.com/pparys/ho-transform-sbs}{https://github.com/pparys/ho-transform-sbs}).
	The strategy of our proof is similar as in Asada and Kobayashi~\cite{word2tree}.
	Namely, we first show that reductions performed by $\Gg$ can be reordered, so that we can postpone substituting for (trailing) variables of order $0$.
	To store such postponed substitutions, called \emph{explicit substitutions}, we introduce \emph{extended terms}.
	Then, we show that such reordered reductions in $\Gg$ are in a direct correspondence with reductions in $\Gg^\dag$.\footnote{%
		Asada and Kobayashi have an additional step in their proof, namely a reduction to the case of recursion-free grammars.
		This step turns out to be redundant, at least in the case of our transformation.}

\subparagraph{Extended terms.}

	In the sequel, terms defined previously are sometimes called non-extended terms, in order to distinguish them from extended terms defined below.
	Having a finite set of typed nonterminals $\Xx$, and a finite set $\Zz$ of variables of type $\otyp$, 
	\emph{extended terms} over $(\Xx,\Zz)$ are defined by induction:
	\begin{itemize}
	\item	if $z\not\in\Zz$ is a variable of type $\otyp$, and $E$ is an extended term over $(\Xx,\Zz\uplus\set{z})$, and $L$ is a non-extended term of type $\otyp$ over $(\Xx,\Zz)$,
		then $\esubst{E}{L}{z}$ is an extended term over $(\Xx,\Zz)$;
	\item	every non-extended term of type $\otyp$ over $(\Xx,\Zz)$ is an extended term over $(\Xx,\Zz)$.
	\end{itemize}
	The construction of the form $\esubst{E}{L}{z}$ is called an \emph{explicit substitution}.
	Intuitively, it denotes the term obtained by substituting $L$ for $z$ in $E$.
	Notice that the variable $z$ being free in $E$ becomes bound in $\esubst{E}{L}{z}$,
	and that explicit substitutions are allowed only for the ground type $\otyp$.
	
	Of course a (non-extended or extended) term over $(\Xx,\Zz)$ can be also seen as a term over $(\Xx,\Zz')$, where $\Zz'\supseteq\Zz$.
	In the sequel, such extending of the set of variables is often performed implicitly.
	
	Having a grammar $\Gg=(\Xx,X_0,\Rr)$, for every set $\Zz$ of variables of type $\otyp$
	we define an \emph{ext-reduction} relation $\erew_\Gg$ between extended terms over $(\Xx,\Zz)$ and sets of such terms, 
	as the least relation such that
	\begin{bracketenumerate}
	\item	$X\,K_1\,\dots\,K_k\,L_1\,\dots\,L_\ell\erew_\Gg\set{\esubstdots{R[K_1/y_1,\dots,K_k/y_k,z_1'/z_1,\dots,z_\ell'/z_\ell]}{L_1}{z_1'}{L_\ell}{z_\ell'}}$
		if $\ell=\gar(X)$, and $\Rr(X)=(X\,y_1\,\dots\,y_k\,z_1\,\dots\,z_\ell\to R)$, and $z_1',\dots,z_\ell'$ are fresh variables of type $\otyp$ not appearing in $\Zz$,
	\item	$\bullet\symb[K_1][\dots][K_k]\erew_\Gg\set{K_1,\dots,K_k}$,
	\item	$\oplus\symb[K_1][\dots][K_k]\erew_\Gg\set{K_i}$ for every $i\in[k]$, 
	\item	$\esubst{z}{L}{z}\erew_\Gg\set{L}$,
	\item	$\esubst{z'}{L}{z}\erew_\Gg\set{z'}$ if $z'\neq z$, and
	\item	$\esubst{E}{L}{z}\erew_\Gg\setof{\esubst{F}{L}{z}}{F\in\Ff}$ whenever $E\erew_\Gg\Ff$.
	\end{bracketenumerate}
	
	We say that an extended term $E$ over $(\Xx,\emptyset)$ is \emph{$\Gg$-ext-convergent} if $E\rew_\Gg\Ff$ for some set $\Ff$ of $\Gg$-ext-convergent extended terms.
	The grammar $\Gg$ is \emph{ext-convergent} if its starting nonterminal $X_0$ is $\Gg$-ext-convergent.

	There is an ``expand'' function from extended terms to non-extended terms, which performs all the explicit substitutions written in front of an extended term:
	\begin{align*}
		\exp(\esubstdots{K}{L_1}{z_1}{L_\ell}{z_\ell})=K[L_1/z_1]\dots[L_\ell/z_\ell].
	\end{align*}
	We also write $\exp(\Ff)$ for $\setof{\exp(F)}{F\in\Ff}$ (where $\Ff$ is a set of extended terms).
	The following \lcnamecref{g2eg}, saying that we can consider ext-convergence instead of convergence, can be proved in a standard way
	(actually, Asada and Kobayashi have a very similar lemma~\cite[Lemma 18]{word2tree}); for completeness we attach a proof in \cref{app:std2ext}.

	\begin{lemma}\label{g2eg}
		Let $\Gg=(\Xx,X_0,\Rr)$ be a grammar.
		An extended term $E$ over $(\Xx,\emptyset)$ is $\Gg$-ext-convergent if and only if $\exp(E)$ is $\Gg$-convergent.
		In particular $\Gg$ is ext-convergent if and only if it is convergent.
	\end{lemma}
	
	We extend the transformation function to extended terms, by adding the following rule, where $\esubst{E}{L}{z}$ is an extended term over $(\Xx,\Zz)$, 
	and $Z\in 2^\Zz$ (the first argument is always $\emptyset$, because all extended terms are of type $\otyp$):
	\begin{bracketenumerate}
	\setcounter{enumi}{8}
	\item	$\tr(\emptyset,Z,\esubst{E}{L}{z}) = \oplus\symb[\tr(\emptyset,Z[z\mapsto0],E)][\bullet\symb[\tr(\emptyset,Z[z\mapsto1],E)][\tr(\emptyset,Z,L)]]$.
	\end{bracketenumerate}

\subparagraph{Between ext-convergence and convergence of $\Gg^\dag$.}

	Once we know that convergence and ext-convergence of $\Gg$ are equivalent (cf.~\cref{g2eg}), it remains to prove that ext-convergence of $\Gg$ is equivalent to convergence of $\Gg^\dag$,
	which is the subject of \cref{eg2trans}:

	\begin{lemma}\label{eg2trans}
		Let $\Gg=(\Xx,X_0,\Rr)$ be a grammar.
		An extended term $E$ over $(\Xx,\emptyset)$ is $\Gg$-ext-convergent if and only if $\tr(\emptyset,\emptyset,E)$ is $\Gg^\dag$-convergent.
		In particular $\Gg$ is ext-convergent if and only if $\Gg^\dag$ is convergent.
	\end{lemma}

	The remaining part of this section is devoted to a proof of this \lcnamecref{eg2trans}.
	Fix a grammar $\Gg=(\Xx,X_0,\Rr)$.
	Of course the second part (concerning the grammars)
	follows from the first part (concerning an extended term) applied to the starting nonterminal $X_0$.
	It is thus enough to prove the first part.
	We start with the left-to-right direction (i.e., from $\Gg$-ext-convergence of $E$ to $\Gg^\dag$-convergence of $\tr(\emptyset,\emptyset,E)$).
	We need two simple auxiliary \lcnamecrefs{trans-subst-com}.
	The first of them says that the $\tr$ function commutes with substitution:

	\begin{lemma}\label{trans-subst-com}
		Let $R[K_1/y_1,\dots,K_k/y_k]$ be a term over $(\Xx,\Zz)$, let $A\in2^{[\gar(R)]}$, and let $Z\in2^\Zz$.
		Then
		\begin{align*}
			\tr(A,Z,R[K_1/y_1,\dots,K_k/y_k]) = (\tr(A,Z,R))[\tr(B,Z,K_i)/y^\dag_{i,B}\mid i\in[k],B\in2^{[\gar(K_i)]}].
		\end{align*}
	\end{lemma}
	
	\begin{proof}
		A straightforward induction on the structure of $R$.
	\end{proof}
	
	The second \lcnamecref{greater-reduces-more} says that by increasing values of the function $Z$ we can make the transformed term only more convergent:
	
	\begin{lemma}\label{greater-reduces-more}
		Let $E$ be an extended term over $(\Xx,\Zz\uplus\set{z})$, and let $Z\in2^\Zz$.
		If $\tr(\emptyset,\allowbreak Z[z\mapsto\nobreak0],\allowbreak E)$ is $\Gg^\dag$-convergent, then also $\tr(\emptyset,Z[z\mapsto1],E)$ is $\Gg^\dag$-convergent.
	\end{lemma}
	
	\begin{proof}
		Denote $P^0=\tr(\emptyset,Z[z\mapsto0],E)$ and $P^1=\tr(\emptyset,Z[z\mapsto1],E)$.
		Tracing the rules of the transformation function, we can see that $P^0$ and $P^1$ are created in the same way,
		with the exception that occurrences of $z$ in $E$ are transformed to $\Omega$ in $P^0$, and to $\bullet$ in $P^1$.
		Thus, $P^1$ can be obtained from $P^0$ by replacing some occurrences of $\Omega$ to $\bullet$.
		We know that $P^0$ is $\Gg^\dag$-convergent, which means that it can be rewritten using the $\rew_\Gg$ relation until reaching empty sets.
		Moreover, the subterms $\Omega$ (which are present in $P^0$, but not in $P^1$) cannot be reached during this rewriting, because $\Omega$ is not $\Gg^\dag$-convergent.
		Thus, $P^1$ can be rewritten in exactly the same way as $P^0$, so it is also $\Gg^\dag$-convergent.
	\end{proof}

	The next \lcnamecref{lem:ext-reduces2trans-reduces} shows how ext-reductions of $\Gg$ are reflected in $\Gg^\dag$:
	
	\begin{lemma}\label{lem:ext-reduces2trans-reduces}
		Let $E$ be an extended term over $(\Xx,\Zz)$ and let $Z\in2^\Zz$.
		If $E\erew_\Gg\Ff$ and $\tr(\emptyset,Z,F)$ is $\Gg^\dag$-convergent for every $F\in\Ff$, then $\tr(\emptyset,Z,E)$ is also $\Gg^\dag$-convergent.
	\end{lemma}
	
	\begin{proof}
		Induction on the definition of $E\erew_\Gg\Ff$.
		We analyze particular cases appearing in the definition.
		Missing details are given in \cref{app:ext-reduces2trans-reduces}.
		
		In Case (1) $E$ consists of an application of arguments to some nonterminal $X$.
		For simplicity of presentation, suppose that $X$ has two arguments: $y$ of positive order, and $z$ of order $0$
		(the general case is handled in the appendix).
		Then 
		\begin{align*}
			E=X\,K\,L,&&
			\mbox{and}&&
			\Ff=\set{F}&&
			\mbox{for}&&
			F=\esubst{R[K/y,z'/z]}{L}{z'},
		\end{align*}
		where $\Rr(X)=(X\,y\,z\to R)$ and $z'$ is a fresh variable of type $\otyp$ not appearing in $\Zz$.
		For $j\in\set{0,1}$ let
		\begin{align*}
			P^j&=\tr([1\mapsto j],Z,X\,K),&
			\mbox{and}&&
			Q^j&=\tr(\emptyset,Z[z'\mapsto j],R[K/y,z'/z]).
		\end{align*}
		First, we prove that $P^j\rew_{\Gg^\dag}\set{Q^j}$.
		By definition we have that
		\begin{align*}
			P^j=X^\dag_{[1\mapsto j]}\,(\tr(B,Z,K))_{B\in2^{[\gar(K)]}},
		\end{align*}
		and by \cref{trans-subst-com} we have that
		\begin{align*}
			Q^j&=\tr(\emptyset,Z[z'\mapsto j],R[z'/z])[\tr(B,Z[z'\mapsto j],K)/y^\dag_B\mid B\in2^{[\gar(K)]}]\\
			&=\tr(\emptyset,[z\mapsto j],R)[\tr(B,Z,K)/y^\dag_B\mid B\in2^{[\gar(K)]}]],
		\end{align*}
		where the second equality holds because the $z'$ does not appear in $K$ and the variables from $\dom(Z)$ do not appear in $R$.
		Recalling that the rule for $X_A^\dag$ is
		\begin{align*}
			X_{[1\mapsto j]}^\dag\,(y_B^\dag)_{B\in2^{[\gar(y)]}}\to\tr(\emptyset,[z\mapsto j],R),
		\end{align*}
		we immediately see that indeed $P^j\rew_{\Gg^\dag}\set{Q^j}$.
		Having this, we recall that
		\begin{align}\label[equalities]{eq:EF}
			\tr(\emptyset,Z,E)=\oplus\symb[P^0][\bullet\symb[P^1][L']]&&
			\mbox{and}&&
			\tr(\emptyset,Z,F)=\oplus\symb[Q^0][\bullet\symb[Q^1][L']]
		\end{align}
		for appropriate $L'$ (obtained by transforming $L$).
		Recall that, by definition, a term $M$ is $\Gg^\dag$-convergent if and only if $M\rew_{\Gg^\dag}\Nn$ for some set $\Nn$ of $\Gg^\dag$-convergent terms.
		Thus, the only way why $\tr(\emptyset,Z,F)$ can be $\Gg^\dag$-convergent (which holds by assumption) is that  
		either $Q^0$ is $\Gg^\dag$-convergent, or both $Q^1$ and $L'$ are $\Gg^\dag$-convergent.
		Because of the reduction $P^j\rew_{\Gg^\dag}\set{Q^j}$ we have that either $P^0$ is $\Gg^\dag$-convergent, or both $P^1$ and $L'$ are $\Gg^\dag$-convergent,
		which implies that $\tr(\emptyset,Z,E)$ is $\Gg^\dag$-convergent.

		In Cases (2) and (3), when $E=\bullet\symb[K_1][\dots][K_k]$ or $E=\oplus\symb[K_1][\dots][K_k]$,
		we have a reduction from $\tr(\emptyset,Z,E)$ to $\set{\tr(\emptyset,Z,F)\mid F\in\Ff}$,
		because $\tr$ distributes over $\bullet\symb[\dots]$ and $\oplus\symb[\dots]$.
		In Cases (4) and (5) (elimination of explicit substitution) we also have similar reductions.

		Finally, in Case (6) we have that
		\begin{align*}
			E&=\esubst{E_0}{L}{z}, &
			\Ff&=\set{\esubst{E_1}{L}{z},\dots,\esubst{E_k}{L}{z}},&&\mbox{and}
			&E_0\erew_\Gg\set{E_1,\dots,E_k}.
		\end{align*}
		By definition, for every $i\in\set{0,\dots,k}$ we have that
		\begin{align}
			&\tr(\emptyset,Z,\esubst{E_i}{L}{z})=\oplus\symb[P_i^0][\bullet\symb[P_i^1][L']],\qquad\mbox{where}\displaybreak[0]\label{eq:EP}\\
			&P_i^0=\tr(\emptyset,Z[z\mapsto0],E_i),\qquad
			P_i^1=\tr(\emptyset,Z[z\mapsto1],E_i),\qquad
			L'=\tr(\emptyset,Z,L).\nonumber
		\end{align}
		Thus, $\tr(\emptyset,Z,\esubst{E_i}{L}{z})$ is $\Gg^\dag$-convergent if and only if
		either $P_i^0$ is $\Gg^\dag$-convergent, or both $P_i^1$ and $L'$ are $\Gg^\dag$-convergent.
		By assumption this is the case for all $i\in[k]$, and we have to prove this for $i=0$.
		If for every $i\in[k]$ we have the former case (i.e., $P_i^0$ is $\Gg^\dag$-convergent),
		by the induction hypothesis (used with the function $Z[z\mapsto0]$) we have that $P_0^0$ is $\Gg^\dag$-convergent, and we are done.
		In the opposite case, for some $i\in[k]$ (but for at least one of them) we have that both $P_i^1$ and $L'$ are $\Gg^\dag$-convergent,
		and for the remaining $i\in[k]$ we have that $P_i^0$ is $\Gg^\dag$-convergent.
		Using \cref{greater-reduces-more} we deduce that if $P_i^0$ is $\Gg^\dag$-convergent, then also $P_i^1$ is $\Gg^\dag$-convergent.
		Thus actually $P_i^1$ is $\Gg^\dag$-convergent for every $i\in[k]$, and additionally $L'$ is $\Gg^\dag$-convergent.
		By the induction hypothesis (used with the function $Z[z\mapsto1]$) we have that $P_0^1$ is $\Gg^\dag$-convergent, and we are also done.
	\end{proof}

	We can now conclude with the left-to-right direction of \cref{eg2trans}:

	\begin{lemma}
		Let $E$ be an extended term over $(\Xx,\emptyset)$.
		If $E$ is $\Gg$-ext-convergent, then $\tr(\emptyset,\emptyset,E)$ is $\Gg^\dag$-convergent.
	\end{lemma}
	
	\begin{proof}
		Induction on the fact that $E$ is $\Gg$-ext-convergent.
		Because $E$ is $\Gg$-ext-convergent, $E\erew_\Gg\Ff$ for some set $\Ff$ of $\Gg$-ext-convergent extended terms, for which we can apply the induction hypothesis.
		The induction hypothesis says that $\tr(\emptyset,\emptyset,F)$ is $\Gg^\dag$-convergent for every $F\in\Ff$.
		In such a situation \cref{lem:ext-reduces2trans-reduces} implies that $\tr(\emptyset,\emptyset,E)$ is also $\Gg^\dag$-convergent, as required.
	\end{proof}

	For a proof in the opposite direction we need the following definition.
	We say that a term $M$ \emph{$\Gg^\dag$-convergent in $n$ steps} if $M\rew_{\Gg^\dag}\set{N_1,\dots,N_k}$, and every $N_i$ is $\Gg^\dag$-convergent in $n_i$ steps, 
	and $n=1+n_1+\dots+n_k$ (i.e., we count $1$ for the above reduction, and we sum the numbers of steps needed to reduce all $N_i$).
	Clearly a term $M$ is $\Gg^\dag$-convergent if and only if it is $\Gg^\dag$-convergent in $n$ steps for some $n\in\Nat$.
	Notice that the number $n$ is not determined by $M$ (i.e., that the same term $M$ can be $\Gg^\dag$-convergent in $n$ steps for multiple values of $n$).
	We can now state the converse of \cref{lem:ext-reduces2trans-reduces}:

	\begin{lemma}\label{lem:trans-reduces2ext-reduces} 
		Let $E$ be an extended term over $(\Xx,\Zz)$ and let $Z\in2^\Zz$.
		If $\tr(\emptyset,Z,E)$ is $\Gg^\dag$-convergent in $n$ steps and $E$ is not a variable, 
		then there exists a set $\Ff$ of extended terms such that
		$E\erew_\Gg\Ff$ and $\tr(\emptyset,Z,F)$ is $\Gg^\dag$-convergent in less than $n$ steps for every $F\in\Ff$.
	\end{lemma}
	
	\begin{proof}
		Induction on the number of explicit substitutions in $E$.
		Depending on the shape of $E$, we have several cases.
		Missing details are given in \cref{app:trans-reduces2ext-reduces}.
		
		One case is $E$ consists of a nonterminal $X$ with some arguments applied.
		For simplicity of presentation, we again suppose that $X$ has two arguments: $y$ of positive order, and $z$ of order $0$.
		Thus, $E$ is of the form $E=X\,K\,L$.
		Let $X\,y\,z\to R$ be the rule for $X$, and let $z'$ be a fresh variable of type $\otyp$ not appearing in $\Zz$.
		In such a situation, taking $F=\esubst{R[K/y,z'/z]}{L}{z'}$ we have that $E\erew_\Gg\set{F}$.
		Recall the terms $P^j$ and $Q^j$ (for $j\in\set{0,1}$) from the proof of \cref{lem:ext-reduces2trans-reduces}.
		In that proof we have observed that $P^j\rew_{\Gg^\dag}\set{Q^j}$.
		But clearly this is the only way how $P^j$ can reduce,
		so if $P^j$ is $\Gg^\dag$-convergent in $n_j$ steps, then necessarily $Q^j$ is $\Gg^\dag$-convergent in $n_j-1$ steps.
		By \cref{eq:EF} we have that if $\tr(\emptyset,Z,E)$ is $\Gg^\dag$-convergent in $n$ steps,
		then either $P^0$ is $\Gg^\dag$-convergent in $n_0=n-1$ steps,
		or both $P^1$ and $L'$ are $\Gg^\dag$-convergent in, respectively, $n_1$ and $n-n_1-2$ steps, for some $n_1\in\Nat$.
		In the former case, $Q^0$ is $\Gg^\dag$-convergent in $n_0-1=n-2$ steps, so $\tr(\emptyset,Z,F)$ is $\Gg^\dag$-convergent in $n-1$ steps, and we are done.
		In the latter case, $Q^1$ is $\Gg^\dag$-convergent in $n_1-1$ steps, so $\tr(\emptyset,Z,F)$ is $\Gg^\dag$-convergent in $(n_1-1)+(n-n_1-2)+2=n-1$ steps, and we are done again.
		
		Notice that we do not have a similar case for a variable with some arguments applied, because the whole $E$ is not a variable,
		and because (by definition of an extended term) all free variables of $E$ are of type $\otyp$.
		
		The cases of $E=\bullet\symb[K_1][\dots][K_k]$ and $E=\oplus\symb[K_1][\dots][K_k]$ are straightforward.

		It remains to assume that $E$ is an explicit substitution.
		If $E=\esubst{z}{L}{z}$, we should take $\Ff=\set{L}$, and if $E=\esubst{z'}{L}{z}$ for $z'\neq z$, we should take $\Ff=\set{z'}$
		(in these two subcases we cannot use the induction assumption, because it does not work for an extended term being a single variable).
		Otherwise $E=\esubst{E_0}{L}{z}$, where $E_0$ is not a variable.
		Recall that $\tr(\emptyset,Z,E)=\oplus\symb[P_0^0][\bullet\symb[P_0^1][L']]$ for $P_0^0, P_0^1, L'$ as in the proof of \cref{lem:ext-reduces2trans-reduces}.
		By assumption $\tr(\emptyset,Z,E)$ is $\Gg^\dag$-convergent in $n$ steps, so either $P_0^0$ is $\Gg^\dag$-convergent in $n'=n-1$ steps,
		or both $P_0^1$ and $L'$ are $\Gg^\dag$-convergent in, respectively, $n'$ and $n-n'-2$ steps, for some $n'\in\Nat$.
		Let $j=0$ in the former case and $j=1$ in the latter case.
		The induction hypothesis gives us a set $\set{E_1,\dots,E_k}$ such that $E_0\erew_\Gg\set{E_1,\dots,E_k}$
		and $\tr(\emptyset,Z[z\mapsto j],E_i)$ is $\Gg^\dag$-convergent in less than $n'$ steps for every $i\in[k]$.
		We then take
		\begin{align*}
			\Ff&=\set{\esubst{E_1}{L}{z},\dots,\esubst{E_k}{L}{z}}.
		\end{align*}
		\cref{eq:EP} holds now for all $i\in\set{0,\dots,k}$.
		For $j=0$ we use that the fact that $\tr(\emptyset,Z,\esubst{E_i}{L}{z})\rew_{\Gg^\dag}\set{P_i^0}$, 
		which implies that $\tr(\emptyset,Z,\esubst{E_i}{L}{z})$ is $\Gg^\dag$-convergent in less than $n'+1=n$ steps, as required.
		For $j=1$ we use that the fact that $\tr(\emptyset,Z,\esubst{E_i}{L}{z})\rew_{\Gg^\dag}\set{\bullet\symb[P_i^1][L']}$ 
		and $\bullet\symb[P_i^1][L']\rew_{\Gg^\dag}\set{P_i^1,L'}$,
		which implies that $\tr(\emptyset,Z,\esubst{E_i}{L}{z})$ is $\Gg^\dag$-con\-ver\-gent in less than $n'+(n-n'-2)+2=n$ steps, as required.
	\end{proof}

	The next \lcnamecref{last-lemma} finishes the proof of \cref{eg2trans}, and thus the proof of correctness of our transformation:

	\begin{lemma}\label{last-lemma}
		Let $E$ be an extended term over $(\Xx,\emptyset)$.
		If $\tr(\emptyset,\emptyset,E)$ is $\Gg^\dag$-convergent then $E$ is $\Gg$-ext-convergent.
	\end{lemma}
	
	\begin{proof}
		Induction on the (smallest) number $n$ such that $\tr(\emptyset,\emptyset,E)$ is $\Gg^\dag$-convergent in $n$ steps.
		By assumption $E$ is not a variable, because it is an extended term over $(\Xx,\emptyset)$ (no free variables).
		So, by \cref{lem:ext-reduces2trans-reduces} there exists a set $\Ff$ of extended terms such that
		$E\erew_\Gg\Ff$ and $\tr(\emptyset,\emptyset,F)$ is $\Gg^\dag$-convergent in less than $n$ steps for every $F\in\Ff$.
		By the induction hypothesis every $F\in\Ff$ is $\Gg$-ext-convergent, so by definition also $E$ is $\Gg$-ext-convergent.
	\end{proof}

\section{Conclusions}

	We have presented a new, simple algorithm checking whether a higher-order grammar generates a nonempty language.
	One may ask whether this algorithm can be used in practice.
	Of course the complexity $n$\textsf{-EXPTIME} for grammars of order $n$ is unacceptably large
	(even if we take into account the fact that we are $n$-fold exponential only in the arity of types, not in the size of a grammar),
	but one has to recall that there exist tools solving the considered problem in such a complexity.
	The reason why these tools work is that the time spent by them on ``easy'' inputs is much smaller than the worst-case complexity (and many ``typical inputs'' are indeed easy).
	Unfortunately, this is not the case for our algorithm: the size of the grammar resulting from our transformation is always large,
	even if the original grammar generated a nonempty (or empty) language for some ``easy reason''.
	Thus, our algorithm is mainly of a theoretical interest.
	
	The presented transformation preserves nonemptiness, and thus can be used to solve the nonemptiness problem for higher-order grammars.
	However, it seems feasible that other problems concerning higher-order grammars (higher-order recursion schemes),
	like model-checking against parity automata or the simultaneous unboundedness problem~\cite{diagonal}, can be solved using similar transformations.
	Developing such transformations is a possible direction for further work.

\bibliography{bib}

\appendix

\section{Proof of Lemma~\ref{g2eg}}\label{app:std2ext}

	In this section we prove \cref{g2eg}.
	To this end, fix a grammar $\Gg=(\Xx,X_0,\Rr)$.	
	Of course the second part (about convergence/ext-convergence of the grammar)
	follows from the first part (about convergence/ext-convergence of an extended term) applied to the starting nonterminal $X_0$.
	It is thus enough to prove the first part.
	We start with the left-to-right direction (i.e., from ext-convergence to convergence).
	We need a simple auxiliary \lcnamecref{lem:Redudes-add-SubstituteFirst}:

	\begin{lemma}\label{lem:Redudes-add-SubstituteFirst}
		If $M\rew_\Gg\Nn$ and $M[K/x]$ is a valid term, then $M[K/x]\rew_\Gg\setof{N[K/x]}{N\in\Nn}$.
	\end{lemma}
	
	\begin{proof}
		 A trivial case analysis on the definition of $M\rew_\Gg\Nn$.
	\end{proof}
	
	The main ingredient of the proof is the following \lcnamecref{lem:e2ne-transfer}:
	
	\begin{lemma}\label{lem:e2ne-transfer}
		If $E\erew_\Gg\Ff$, then either $\exp(E)\rew_\Gg\exp(\Ff)$, or $\exp(\Ff)=\set{\exp(E)}$.
	\end{lemma}
	
	\begin{proof}
		Induction on the definition of $E\erew_\Gg\Ff$.
		\begin{bracketenumerate}
		\item	If 
			\begin{align*}
				E&=X\,K_1\,\dots\,K_k\,L_1\,\dots\,L_\ell &&\mbox{and}\\
				\Ff&=\set{\esubstdots{R[K_1/y_1,\dots,K_k/y_k,z_1'/z_1,\dots,z_\ell'/z_\ell]}{L_1}{z_1'}{L_\ell}{z_\ell'}}, 
			\end{align*}
			where $\ell=\gar(X)$, then
			\begin{align*}
				\exp(E)=E&\rew_\Gg\set{R[K_1/y_1,\dots,K_k/y_k,L_1/z_1,\dots,L_\ell/z_\ell]}\\
				&=\set{\exp(\esubstdots{R[K_1/y_1,\dots,K_k/y_k,z_1'/z_1,\dots,z_\ell'/z_\ell]}{L_1}{z_1'}{L_\ell}{z_\ell'})}\\
				&=\exp(\Ff).
			\end{align*}
		\item	If $E=\bullet\symb[K_1][\dots][K_k]$ and $\Ff=\set{K_1,\dots,K_k}$, then
			\begin{align*}
				\exp(E)=E\rew_G\set{K_1,\dots,K_k}=\set{\exp(K_1),\dots,\exp(K_k)}=\exp(\Ff).
			\end{align*}
		\item	If $E=\oplus\symb[K_1][\dots][K_k]$ and $\Ff=\set{K_i}$, then
			\begin{align*}
				\exp(E)=E\rew_G\set{K_i}=\set{\exp(K_i)}=\exp(\Ff).
			\end{align*}
		\item	If $E=\esubst{z}{L}{z}$ and $\Ff=\set{L}$, then $\exp(\Ff)=\set{\exp(L)}=\set{L}=\set{\exp(E)}$.
		\item	If $E=\esubst{z'}{L}{z}$ and $\Ff=\set{z'}$, then $\exp(\Ff)=\set{\exp(z')}=\set{z'}=\set{\exp(E)}$.
		\item	Finally, if $E=\esubst{E'}{L}{z}$ and $\Ff=\setof{\esubst{F'}{L}{z}}{F'\in\Ff'}$ and $E'\erew_\Gg\Ff'$,
			we have $\exp(E)=(\exp(E'))[L/z]$ and $\exp(\Ff)=\setof{(\exp(F'))[L/z]}{F'\in\Ff'}=\setof{N[L/z]}{N\in\exp(\Ff')}$;
			by the induction hypothesis we have that either $\exp(E')\rew_\Gg\exp(\Ff')$, or $\exp(\Ff')=\set{\exp(E')}$;
			the latter immediately implies that $\exp(\Ff)=\set{\exp(E)}$, while the former implies $\exp(E)\rew_\Gg\exp(\Ff)$ by \cref{lem:Redudes-add-SubstituteFirst}.
		\qedhere\end{bracketenumerate}
	\end{proof}

	We can now conclude with the left-to-right direction of \cref{g2eg}:
	
	\begin{lemma}\label{lem:r2l}
		If an extended term $E$ over $(\Xx,\emptyset)$ is $\Gg$-ext-convergent, then $\exp(E)$ is $\Gg$-convergent.
	\end{lemma}
	
	\begin{proof}
		Induction on the fact that $E$ is $\Gg$-ext-convergent.
		Because $E$ is $\Gg$-ext-convergent, $E\erew_\Gg\Ff$ for some set $\Ff$ of $\Gg$-ext-convergent extended terms (for which we can apply the induction hypothesis).
		Using the induction hypothesis, for every extended term $F$ in $\Ff$ we obtain that $\exp(F)$ is $\Gg$-convergent,
		that is, all terms in $\exp(\Ff)$ are $\Gg$-convergent.
		By \cref{lem:e2ne-transfer}, we have that either $\exp(E)\rew_\Gg\exp(\Ff)$, or $\exp(\Ff)=\set{\exp(E)}$.
		In the latter case we already know that $\exp(E)$ (as an element of $\exp(\Ff)$) is $\Gg$-convergent.
		In the former case, we use the definition of $\Gg$-convergence, and we also obtain that $\exp(E)$ is $\Gg$-convergent.
	\end{proof}

	We now come to the opposite direction: from convergence to ext-convergence.
	We say that an extended term $E$ is \emph{simplified} if $E$ is not of the form $\esubstdots{z}{L_1}{z_1}{L_\ell}{z_\ell}$
	(i.e., if the non-extended term inside all explicit substitutions is not a variable).
	It turns out that every extended term can be ext-reduced to a simplified one (this is shown in the proof of \cref{lem:l2r}).
	Then, for a simplified extended term we can find an ext-reduction corresponding to a given standard reduction:
	
	\begin{lemma}\label{lem:simpl-transfer}
		If $E$ is a simplified extended term over $(\Xx,\emptyset)$, and if $\exp(E)\rew_\Gg\Nn$, 
		then $E\erew_\Gg\Ff$ for some $\Ff$ such that $\exp(\Ff)=\Nn$.
	\end{lemma}
	
	\begin{proof}
		The extended term $E$ can be written in the form $E=\esubstdots{E'}{R_1}{z_1'}{R_s}{z_s'}$, where $E'$ is a non-extended term.
		For every term $K$ let $\Xi(K)=K[R_1/z_1']\dots[R_s/z_s']$; in particular $\exp(E)=\Xi(E')$.
		Below, we prove that there exists $\Ff'$ such that $E'\erew_\Gg\Ff'$ and $\setof{\Xi(\exp(F'))}{F'\in\Ff'}=\Nn$.
		This already gives the thesis for $\Ff=\setof{\esubstdots{F'}{R_1}{z_1'}{R_s}{z_s'}}{F'\in\Ff'}$.
		Indeed, on the one hand, because of $E'\erew_\Gg\Ff'$, due to rule (6) of the definition of $\erew_\Gg$ (applied $s$ times), we have that
		\begin{align*}
			E=\esubstdots{E'}{R_1}{z_1'}{R_s}{z_s'}\erew_\Gg\setof{\esubstdots{F'}{R_1}{z_1'}{R_s}{z_s'}}{F'\in\Ff'}=\Ff.
		\end{align*}
		On the other hand, due to $\setof{\Xi(\exp(F'))}{F'\in\Ff'}=\Nn$, we have that 
		\begin{align*}
			\exp(\Ff)&=\setof{\exp(\esubstdots{F'}{R_1}{z_1'}{R_s}{z_s'})}{F'\in\Ff'}=\setof{\Xi(\exp(F'))}{F'\in\Ff'}=\Nn.
		\end{align*}
		
		Thus, it remains to prove existence of the aforementioned set $\Ff'$.
		We have four cases depending on the shape of $E'$:
		\begin{bracketenumerate}\setcounter{enumi}{-1}
		\item   The term $E'$ is of the form $y\,M_1\,\dots\,M_m$ for some variable $y$.
			Then $y$ has to be one of the variables $z_1',\dots,z_s'$, because the whole $E$ has no free variables;
			in particular $z$ is of order $0$.
			However, this is actually impossible.
			Indeed, for $m=0$ this is impossible, because $E$ is simplified (i.e., $E'$ is not a variable),
			and for $m\geq 1$ this is impossible, because then $y$ would be of positive order.
		\item   The term $E'$ is of the form $X\,K_1\,\dots\,K_k\,L_1\,\dots\,L_\ell$ for a nonterminal $X$, where $\ell=\gar(X)$.
			Let $X\,y_1\,\dots\,y_k\,z_1\,\dots\,z_\ell\to R$ be the rule for $X$.
			Then $\exp(E)=X\,(\Xi(K_1))\,\dots\,(\Xi(K_k))\allowbreak\,(\Xi(L_1))\,\dots\,(\Xi(L_\ell))$,
			so $\exp(E)\rew_\Gg\Nn$ implies that
			\begin{align*}
				\Nn&=\set{R[\Xi(K_1)/y_1,\dots,\Xi(K_k)/y_k,\Xi(L_1)/z_1,\dots,\Xi(L_\ell)/z_\ell]}\\
				&=\set{\Xi(R[K_1/y_1,\dots,K_k/y_k,L_1/z_1,\dots,L_\ell/z_\ell]},
			\end{align*}
			where the second equality holds because $R$ does not contain the variables $z_1',\dots,z_s'$.
			We take $\Ff'=\set{\esubstdots{R[K_1/y_1,\dots,K_k/y_k,z_{s+1}'/z_1,\dots,z_{s+\ell}'/z_\ell]}{L_1}{z_{s+1}'}{L_\ell}{z_{s+\ell}'}}$.
			Then $E'\erew_\Gg\Ff'$ by rule (1) of the definition of $\erew_\Gg$.
			Simultaneously $\setof{\Xi(\exp(F'))}{F'\in\Ff'}=\set{\Xi(R[K_1/y_1,\dots,K_k/y_k,L_1/z_1,\dots,L_\ell/z_\ell])}=\Nn$.
		\item   The term $E'$ is of the form $\bullet\symb[K_1][\dots][K_k]$.
			Then $\exp(E)=\bullet\symb[\Xi(K_1)][\dots][\Xi(K_k)]$, so $\exp(E)\rew_\Gg\Nn$ implies that $\Nn=\set{\Xi(K_1),\dots,\Xi(K_k)}$.
			We take $\Ff'=\set{K_1,\dots,K_k}$.
			Then $E'\erew_\Gg\Ff'$ by rule (2) of the definition of $\erew_\Gg$.
			Simultaneously $\exp(K_i)=K_i$ for all $i\in[k]$ (the terms $K_i$ are non-extended), so
			$\setof{\Xi(\exp(F'))}{F'\in\Ff'}=\set{\Xi(K_1),\dots,\Xi(K_k)}=\Nn$.
		\item   The term $E'$ is of the form $\oplus\symb[K_1][\dots][K_k]$.
			Then $\exp(E)=\oplus\symb[\Xi(K_1)][\dots][\Xi(K_k)]$, so $\exp(E)\rew_\Gg\Nn$ implies that $\Nn=\set{\Xi(K_i)}$ for some $i\in[k]$.
			We take $\Ff'=\set{K_i}$.
			Then $E'\erew_\Gg\Ff'$ by rule (3) of the definition of $\erew_\Gg$.
			Simultaneously $\exp(K_i)=K_i$, so
			$\setof{\Xi(\exp(F'))}{F'\in\Ff'}=\set{\Xi(K_i)}=\Nn$.
		\qedhere\end{bracketenumerate}
	\end{proof}

	In the last \lcnamecref{lem:l2r} we prove the right-to-left direction of \cref{g2eg}:

	\begin{lemma}\label{lem:l2r}
		If $\exp(E)$ is $\Gg$-convergent, for an extended term $E$ over $(\Xx,\emptyset)$, 
		then $E$ is $\Gg$-ext-convergent.
	\end{lemma}
	
	\begin{proof}
		Induction on the fact that $\exp(E)$ is $\Gg$-convergent, and internally on the number of explicit substitutions in $E$.
		One case is that $E$ is simplified.
		Because $\exp(E)$ is $\Gg$-convergent, $\exp(E)\rew_\Gg\Nn$ for some set $\Nn$ of $\Gg$-convergent terms (for which we can apply the induction hypothesis).
		By \cref{lem:simpl-transfer}, there is a set $\Ff$ such that $E\erew_\Gg\Ff$ and $\exp(\Ff)=\Nn$.
		The latter means that $\Nn=\setof{\exp(F)}{F\in\Ff}$.
		Using the induction hypothesis for every term in $\Nn$, we obtain that all extended terms $F$ in $\Ff$ are $\Gg$-ext-convergent.
		Due to $E\erew_\Gg\Ff$, this implies that $E$ is $\Gg$-ext-convergent.
		
		The opposite case is that the extended term $E$ is not simplified, that is, it is of the form $E=\esubstdots{\esubst{z}{L_1}{z_1}}{L_2}{z_2}{L_k}{z_k}$, 
		where $z$ is one of the variables $z_1,\dots,z_k$ (the whole $E$ does not have free variables).
		Suppose first that $z=z_1$, and take $F=\esubstdots{L_1}{L_2}{z_2}{L_k}{z_k}$.
		Then
		\begin{align*}
			\exp(E)=z[L_1/z_1][L_2/z_2]\dots[L_k/z_k]=L_1[L_2/z_2]\dots[L_k/z_k]=\exp(F).
		\end{align*}
		Notice that $F$ has less explicit substitutions than $E$ (the term $L_1$ is non-extended),
		so we can use the internal induction hypothesis, obtaining that $F$ is $\Gg$-ext-convergent.
		Moreover, we have that $\esubst{z}{L_1}{z_1}\erew_\Gg\set{L_1}$ by rule (4) of the definition of $\erew_\Gg$,
		thus also $E\erew_\Gg\set{F}$ by rule (6) of this definition (used $k-1$ times).
		In consequence, also $E$ is $\Gg$-ext-convergent.
		
		When $z=z_i$ for $i\geq 2$, we proceed similarly. Taking $F=\esubstdots{z}{L_2}{z_2}{L_k}{z_k}$
		we have that
		\begin{align*}
			\exp(E)=z[L_1/z_1][L_2/z_2]\dots[L_k/z_k]=z[L_2/z_2]\dots[L_k/z_k]=\exp(F).
		\end{align*}
		Because $F$ has less explicit substitutions than $E$,
		we can use the internal induction hypothesis, obtaining that $F$ is $\Gg$-ext-convergent.
		Moreover, we have that $\esubst{z}{L_1}{z_1}\erew_\Gg\set{z}$ by rule (5) of the definition of $\erew_\Gg$,
		thus also $E\erew_\Gg\set{F}$ by rule (6) of this definition (used $k-1$ times).
		In consequence, also $E$ is $\Gg$-ext-convergent.
	\end{proof}

\section{Additional details for the proof of Lemma~\ref{lem:ext-reduces2trans-reduces}}\label{app:ext-reduces2trans-reduces}

	We now complement the proof of \cref{lem:ext-reduces2trans-reduces} with missing details.
	Recall that we are given an extended term $E$ over $(\Xx,\Zz)$ and a function $Z\in2^\Zz$.
	Knowing that $E\erew_\Gg\Ff$ and that $\tr(\emptyset,Z,F)$ is $\Gg^\dag$-convergent for every $F\in\Ff$, 
	we have to prove that $\tr(\emptyset,Z,E)$ is also $\Gg^\dag$-convergent.

	As already said, we proceed by induction on the definition of $E\erew_\Gg\Ff$, and we analyze particular cases of this definition.

	\begin{bracketenumerate}
	\item	Suppose that 
		\begin{align*}
			E&=X\,K_1\,\dots\,K_k\,L_1\,\dots\,L_\ell &\mbox{and}\\
			\Ff&=\set{\esubstdots{R[K_1/y_1,\dots,K_k/y_k,\allowbreak z_1'/z_1,\dots,z_\ell'/z_\ell]}{L_1}{z_1'}{L_\ell}{z_\ell'}},
		\end{align*}
		where $\ell=\gar(X)$, and $\Rr(X)=(X\,y_1\,\dots\,y_k\,z_1\,\dots\,z_\ell\to R)$, and $z_1',\dots,z_\ell'$ are fresh variables of type $\otyp$ not appearing in $\Zz$.
		For every $s\in\set{0,\dots,\ell}$ and every function $A\in2^{[\ell-s]}$, let
		\begin{align*}
			P_{s,A}&=\tr(A,Z,X\,K_1\,\dots\,K_k\,L_1\,\dots\,L_s),\\
			Z_{s,A}&=Z[z_i'\mapsto A(\ell+1-i)\mid i\in\set{s+1,s+2,\dots,\ell}],&\hspace{-3em}\mbox{and}\\
			Q_{s,A}&=\tr(\emptyset,Z_{s,A},\esubstdots{R[K_1/y_1,\dots,K_k/y_k,z_1'/z_1,\dots,z_\ell'/z_\ell]}{L_1}{z_1'}{L_s}{z_s'}).
		\end{align*}
		We prove, by induction on $s$, that if $Q_{s,A}$ is $\Gg^\dag$-convergent then also $P_{s,A}$ is $\Gg^\dag$-convergent.
		For $s=\ell$ and $A=\emptyset$ this gives the thesis (because $\tr(\emptyset,Z,E)=Q_{\ell,\emptyset}$ and $\setof{\tr(\emptyset,Z,F)}{F\in\Ff}=\set{Q_{\ell,\emptyset}}$).
		
		Suppose first that $s=0$.
		Then
		\begin{align*}
			P_{s,A} &= \tr(A,Z,X\,K_1\,\dots\,K_k)\\
				&= X^\dag_A\,(\tr(B,Z,K_1))_{B\in2^{[\gar(K_1)]}}\,\dots\,(\tr(B,Z,K_k))_{B\in2^{[\gar(K_k)]}},
		\end{align*}
		and, by \cref{trans-subst-com},
		\begin{align*}
			Q_{s,A} &= \tr(\emptyset,Z_{s,A},R[K_1/y_1,\dots,K_k/y_k,z_1'/z_1,\dots,z_\ell'/z_\ell])\\
				&= \tr(\emptyset,Z_{s,A},R[z_1'/z_1,\dots,z_\ell'/z_\ell][K_1/y_1,\dots,K_k/y_k])\\
				&= (\tr(\emptyset,Z_{s,A},R[z_1'/z_1,\dots,z_\ell'/z_\ell]))[\tr(B,Z_{s,A},K_i)/y^\dag_{i,B}\mid i\in[k],B\in2^{[\gar(K_i)]}].
		\end{align*}
		Because the only variables from $\dom(Z_{s,A})$ that appear in $R[z_1'/z_1,\dots,z_\ell'/z_\ell]$ are $z_1',\dots,z_\ell'$, we have that
		\begin{align*}
			\tr(\emptyset,Z_{s,A},R[z_1'/z_1,\dots,z_\ell'/z_\ell])\hspace{-1em}&\\
			&=\tr(\emptyset,[z_i'\mapsto A(\ell+1-i)\mid i\in[\ell]],R[z_1'/z_1,\dots,z_\ell'/z_\ell])\\
			&=\tr(\emptyset,[z_i\mapsto A(\ell+1-i)\mid i\in[\ell]],R).
		\end{align*}
		Likewise, because $z_1',\dots,z_\ell'$ do not appear in $K_1,\dots,K_k$, we have that $\tr(B,Z_{s,A},K_i)=\tr(B,Z,K_i)$ for all $i\in[k]$ and $B\in2^{[\gar(K_i)]}$.
		In consequence,
		\begin{align*}
			Q_{s,A}\hspace{-1em}&\\
			& = (\tr(\emptyset,[z_i\mapsto A(\ell+1-i)\mid i\in[\ell]],R))[\tr(B,Z,K_i)/y^\dag_{i,B}\mid i\in[k],B\in2^{[\gar(K_i)]}].
		\end{align*}
		Recall that the rule for $X_A^\dag$ is
		\begin{align*}
			X_A^\dag\,(y_{1,B}^\dag)_{B\in2^{[\gar(y_1)]}}\,\dots\,(y_{k,B}^\dag)_{B\in2^{[\gar(y_k)]}}\to\tr(\emptyset,[z_i\mapsto A(\ell+1-i)\mid i\in[\ell]],R),
		\end{align*}
		thus $P_{s,A}\rew_{\Gg^\dag}\set{Q_{s,A}}$;
		it follows that if $Q_{s,A}$ is $\Gg^\dag$-convergent then also $P_{s,A}$ is $\Gg^\dag$-convergent.
		
		Next, suppose that $s\geq 1$.
		Let us denote
		\begin{align*}
			P^0&=P_{s-1,A[\ell+1-s\mapsto0]}, & Q^0&=Q_{s-1,A[\ell+1-s\mapsto0]}, & L'=\tr(\emptyset,Z,L_s). \\
			P^1&=P_{s-1,A[\ell+1-s\mapsto1]}, & Q^1&=Q_{s-1,A[\ell+1-s\mapsto1]},
		\end{align*}
		Simultaneously $L'=\tr(\emptyset,Z_{s,A},L_s)$, because variables $z_{s+1}',z'_{s+2},\dots,z_\ell'$ do not appear in $L_s$.
		By definition we have that 
		\begin{align*}
			P_{s,A}&=\oplus\symb[P^0][\bullet\symb[P^1][L']]&&\mbox{and}&
			Q_{s,A}&=\oplus\symb[Q^0][\bullet\symb[Q^1][L']].
		\end{align*}
		Recall that, by definition, a term $M$ is $\Gg^\dag$-convergent if and only if $M\rew_{\Gg^\dag}\Nn$ for some set $\Nn$ of $\Gg^\dag$-convergent terms.
		We can have $Q_{s,A}\rew_{\Gg^\dag}\Nn$ only for $\Nn=\set{Q^0}$ and for $\Nn=\set{\bullet\symb[Q^1][L']}$,
		and we can have $\bullet\symb[Q^1][L']\rew_{\Gg^\dag}\Nn$ only for $\Nn=\set{Q^1,L'}$.
		By assumption $Q_{s,A}$ is $\Gg^\dag$-convergent, so either $Q^0$ is $\Gg^\dag$-convergent, 
		or both $Q^1$ and $L'$ are $\Gg^\dag$-convergent.
		In the former case, $P^0$ is $\Gg^\dag$-convergent by the induction hypothesis, and we have $P_{s,A}\rew_{\Gg^\dag}\set{P^0}$,
		so $P_{s,A}$ is $\Gg^\dag$-convergent, as required.
		In the latter case, $P^1$ is $\Gg^\dag$-convergent by the induction hypothesis,
		and we have $\bullet\symb[P^1][L']\rew_{\Gg^\dag}\set{P^1,L'}$ and $P_{s,A}\rew_{\Gg^\dag}\set{\bullet\symb[P^1][L']}$,
		so also $P_{s,A}$ is $\Gg^\dag$-convergent, as required.
	\item	Suppose that 
		\begin{align*}
			E&=\bullet\symb[K_1][\dots][K_k] &\mbox{and}&&
			\Ff&=\set{K_1,\dots,K_k}.
		\end{align*}
		Then, by definition,
		\begin{align*}
			\tr(\emptyset,Z,E)=\bullet\symb[\tr(\emptyset,Z,K_1)][\dots][\tr(\emptyset,Z,K_k)]\rew_{\Gg^\dag}\set{\tr(\emptyset,Z,K_1),\dots,\tr(\emptyset,Z,K_k)}.
		\end{align*}
		By assumption, elements of the latter set are $\Gg^\dag$-convergent, so also $\tr(\emptyset,Z,E)$ is $\Gg^\dag$-convergent.
	\item	Suppose that 
		\begin{align*}
			E&=\oplus\symb[K_1][\dots][K_k] &\mbox{and}&&
			\Ff&=\set{K_i} &\mbox{for some $i\in[k]$}.
		\end{align*}
		Then, by definition,
		\begin{align*}
			\tr(\emptyset,Z,E)=\oplus\symb[\tr(\emptyset,Z,K_1)][\dots][\tr(\emptyset,Z,K_k)]\rew_{\Gg^\dag}\set{\tr(\emptyset,Z,K_i)}.
		\end{align*}
		By assumption $\tr(\emptyset,Z,K_i)$ is $\Gg^\dag$-convergent, so also $\tr(\emptyset,Z,E)$ is $\Gg^\dag$-convergent.
	\item	Suppose that 
		\begin{align*}
			E&=\esubst{z}{L}{z} &\mbox{and}&&
			\Ff&=\set{L}.
		\end{align*}
		Then 
		\begin{align*}
			&\tr(\emptyset,Z,E)=\oplus\symb[\Omega][\bullet\symb[\bullet][\tr(\emptyset,Z,L)]]\rew_{\Gg^\dag}\set{\bullet\symb[\bullet][\tr(\emptyset,Z,L)]},
			\\
			&{\bullet}\symb[\bullet][\tr(\emptyset,Z,L)]\rew_{\Gg^\dag}\set{\bullet,\tr(\emptyset,Z,L)},
			\qquad\mbox{and}\qquad
			\bullet\rew_{\Gg^\dag}\emptyset.
		\end{align*}
		By assumption $\tr(\emptyset,Z,L)$ is $\Gg^\dag$-convergent, so also $\tr(\emptyset,Z,E)$ is $\Gg^\dag$-convergent.
	\item	Suppose that 
		\begin{align*}
			E&=\esubst{z'}{L}{z} &\mbox{and}&&
			\Ff&=\set{z'},&\mbox{where $z'\neq z$}.
		\end{align*}
		Denote $P=\tr(z',Z,L)$; we simultaneously have $P=\tr(z',Z[z\mapsto 0],L)=\tr(z',\allowbreak Z[z\mapsto\nobreak1],L)$.
		Then 
		\begin{align*}
			&\tr(\emptyset,Z,E)=\oplus\symb[P][\bullet\symb[P][\tr(\emptyset,Z,L)]]\rew_{\Gg^\dag}\set{P}
		\end{align*}
		By assumption $P$ is $\Gg^\dag$-convergent, so also also $\tr(\emptyset,Z,E)$ is $\Gg^\dag$-convergent.
	\item	The last case, when 
		\begin{align*}
			E&=\esubst{E_0}{L}{z}, &
			\Ff&=\set{\esubst{E_1}{L}{z},\dots,\esubst{E_k}{L}{z}},&&\mbox{and}
			&E_0\erew_\Gg\set{E_1,\dots,E_k}
		\end{align*}
		was completely resolved in \cref{sec:correctness}, so we do not repeat the proof here.
	\end{bracketenumerate}

\section{Additional details for the proof of Lemma~\ref{lem:trans-reduces2ext-reduces}}\label{app:trans-reduces2ext-reduces}

	In this section we give missing details for the proof of \cref{lem:trans-reduces2ext-reduces}.
	Recall that we are given an extended term $E$ over $(\Xx,\Zz)$, which is not a variable, and a function $Z\in2^\Zz$.
	Knowing that $\tr(\emptyset,Z,E)$ is $\Gg^\dag$-convergent in $n$ steps,
	we have to prove that there exists a set $\Ff$ of extended terms such that
	$E\erew_\Gg\Ff$ and $\tr(\emptyset,Z,F)$ is $\Gg^\dag$-convergent in less than $n$ steps for every $F\in\Ff$.

	The proof is by induction on the number of explicit substitutions in $E$.
	Depending on the shape of $E$, we have six cases.

	\begin{bracketenumerate}
	\item	Suppose that $E$ starts with an application.
		Then necessarily it can be written as
		\begin{align*}
			E&=X\,K_1\,\dots\,K_k\,L_1\,\dots\,L_\ell,
		\end{align*}
		where $\ell=\gar(X)$.
		In particular, notice that instead of the nonterminal $X$ we cannot have a variable, because (by definition of an extended term) all variables in $\Zz$ are of type $\otyp$.
		Let $X\,y_1\,\dots\,y_k\,z_1\,\dots\,z_\ell\to R$ be the rule for $X$, and let $z_1',\dots,z_\ell'$ be fresh variables of type $\otyp$ not appearing in $\Zz$.
		In such a situation we have that $E\erew_\Gg\set{F}$ for
		\begin{align*}
			F&=\esubstdots{R[K_1/y_1,\dots,K_k/y_k,\allowbreak z_1'/z_1,\dots,z_\ell'/z_\ell]}{L_1}{z_1'}{L_\ell}{z_\ell'}.
		\end{align*}
		For every $s\in\set{0,\dots,\ell}$ and every function $A\in2^{[\ell-s]}$, let us define $P_{s,A}$, $Z_{s,A}$, and $Q_{s,A}$ as in \cref{app:ext-reduces2trans-reduces}.
		We prove, by induction on $s$, that if $P_{s,A}$ is $\Gg^\dag$-convergent in $n_s$ steps (for some $n_s\in\Nat$) then $Q_{s,A}$ is $\Gg^\dag$-convergent in $n_s-1$ steps.
		For $s=\ell$ and $A=\emptyset$ this gives the thesis, taking $\Ff=\set{F}$
		(because $\tr(\emptyset,Z,E)=P_{\ell,\emptyset}$ and $\tr(\emptyset,Z,F)=Q_{\ell,\emptyset}$).
		
		Suppose first that $s=0$.
		Recall that in the proof of \cref{lem:ext-reduces2trans-reduces} we have observed that $P_{s,A}\rew_{\Gg^\dag}\set{Q_{s,A}}$.
		Actually, if $P_{s,A}\rew_{\Gg^\dag}\Nn$ then necessarily $\Nn=\set{Q_{s,A}}$
		(because $P_{s,A}$ is a nonterminal with applied arguments, and in this case the reduction is completely deterministic).
		By assumption $P_{s,A}$ is $\Gg^\dag$-convergent in $n_s$ steps, which, by definition, immediately implies that $Q_{s,A}$ is $\Gg^\dag$-convergent in $n_s-1$ steps, as required.
		
		Next, suppose that $s\geq 1$.
		Using the definition of $P^0$, $P^1$, $Q^0$, $Q^1$, and $L'$ from \cref{app:ext-reduces2trans-reduces},
		we have that 
		\begin{align*}
			P_{s,A}&=\oplus\symb[P^0][\bullet\symb[P^1][L']]&&\mbox{and}&
			Q_{s,A}&=\oplus\symb[Q^0][\bullet\symb[Q^1][L']].
		\end{align*}
		We can have $P_{s,A}\rew_{\Gg^\dag}\Nn$ only for $\Nn=\set{P^0}$ and for $\Nn=\set{\bullet\symb[P^1][L']}$,
		and we can have $\bullet\symb[P^1][L']\rew_{\Gg^\dag}\Nn$ only for $\Nn=\set{P^1,L'}$.
		By assumption $P_{s,A}$ is $\Gg^\dag$-convergent in $n_s$ steps, so either $P^0$ is $\Gg^\dag$-convergent in $n_{s-1}=n_s-1$ steps,
		or both $P^1$ and $L'$ are $\Gg^\dag$-convergent in, respectively, $n_{s-1}$ and $n_s-n_{s-1}-2$ steps, for some $n_{s-1}\in\Nat$.
		In the former case, $Q^0$ is $\Gg^\dag$-convergent in $n_{s-1}-1=n_s-2$ steps by the induction hypothesis, and we have $Q_{s,A}\rew_{\Gg^\dag}\set{Q^0}$,
		so $Q_{s,A}$ is $\Gg^\dag$-convergent in $n_s-1$ steps, as required.
		In the latter case, $Q^1$ is $\Gg^\dag$-convergent in $n_{s-1}-1$ steps by the induction hypothesis,
		and we have $\bullet\symb[Q^1][L']\rew_{\Gg^\dag}\set{Q^1,L'}$ and $Q_{s,A}\rew_{\Gg^\dag}\set{\bullet\symb[Q^1][L']}$,
		so $Q_{s,A}$ is $\Gg^\dag$-convergent in $(n_{s-1}-1)+(n_s-n_{s-1}-2)+2=n_s-1$ steps, as required.
	\item	Suppose that $E$ starts with $\bullet$.
		Then it can be written as $E=\bullet\symb[K_1][\dots][K_k]$, and, by definition,
		\begin{align*}
			\tr(\emptyset,Z,E)=\bullet\symb[\tr(\emptyset,Z,K_1)][\dots][\tr(\emptyset,Z,K_k)],
		\end{align*}
		and $\tr(\emptyset,Z,E)\rew_{\Gg^\dag}\Nn$ only for $\Nn=\set{\tr(\emptyset,Z,K_1),\dots,\tr(\emptyset,Z,K_k)}$.
		We know that $\tr(\emptyset,Z,E)$ is $\Gg^\dag$-convergent in $n$ steps,
		so necessarily the terms $\tr(\emptyset,Z,K_i)$ for $i\in[k]$ are $\Gg^\dag$-convergent in $n_i$ steps, for some numbers $n_i$ such that
		$n_1+\dots+n_k+1=n$.
		In particular all $n_i$ are smaller than $n$, so $\Ff=\set{K_1,\dots,K_k}$ satisfies the thesis, because $E\erew_\Gg\Ff$, by definition.
	\item	Suppose that $E$ starts with $\oplus$.
		Then it can be written as $E=\oplus\symb[K_1][\dots][K_k]$, and, by definition,
		\begin{align*}
			\tr(\emptyset,Z,E)=\oplus\symb[\tr(\emptyset,Z,K_1)][\dots][\tr(\emptyset,Z,K_k)],
		\end{align*}
		and $\tr(\emptyset,Z,E)\rew_{\Gg^\dag}\Nn$ only when $\Nn=\set{\tr(\emptyset,Z,K_i)}$ for some $i\in[k]$.
		We know that $\tr(\emptyset,Z,E)$ is $\Gg^\dag$-convergent in $n$ steps,
		so necessarily $\tr(\emptyset,Z,K_i)$, for some $i\in[k]$, is $\Gg^\dag$-convergent in $n-1$ steps.
		In consequence $\Ff=\set{K_i}$ satisfies the thesis, because $E\erew_\Gg\Ff$, by definition.
	\item	Suppose that $E=\esubst{z}{L}{z}$.
		Then 
		\begin{align*}
			&\tr(\emptyset,Z,E)=\oplus\symb[\Omega][\bullet\symb[\bullet][\tr(\emptyset,Z,L)]].
		\end{align*}
		Notice that $\tr(\emptyset,Z,E)\rew_{\Gg^\dag}\Nn$ only for $\Nn=\set{\Omega}$ and for $\Nn=\set{\bullet\symb[\bullet][\tr(\emptyset,Z,L)]}$.
		In turn, $\Omega$ is not $\Gg^\dag$-convergent (in any number of steps),
		and $\bullet\symb[\bullet][\tr(\emptyset,Z,L)]\rew_{\Gg^\dag}\Nn$ only for $\Nn=\set{\bullet,\tr(\emptyset,Z,L)}$,
		and $\bullet\rew_{\Gg^\dag}\Nn$ only for $\Nn=\emptyset$.
		We know that $\tr(\emptyset,Z,E)$ is $\Gg^\dag$-convergent in $n$ steps, so, by the above, $\tr(\emptyset,Z,L)$ is $\Gg^\dag$-convergent in $n-3$ steps.
		In consequence $\Ff=\set{L}$ satisfies the thesis, because $E\erew_\Gg\Ff$, by definition.
	\item	Suppose that $E=\esubst{z'}{L}{z}$ for some $z'\neq z$.
		Denote $P=\tr(z',Z,L)$; we simultaneously have $P=\tr(z',Z[z\mapsto 0],L)=\tr(z',Z[z\mapsto 1],L)$, so
		\begin{align*}
			&\tr(\emptyset,Z,E)=\oplus\symb[P][\bullet\symb[P][\tr(\emptyset,Z,L)]].
		\end{align*}
		Notice that $\tr(\emptyset,Z,E)\rew_{\Gg^\dag}\Nn$ only for $\Nn=\set{P}$ and for $\Nn=\set{\bullet\symb[P][\tr(\emptyset,Z,L)]}$,
		and $\bullet\symb[P][\tr(\emptyset,Z,L)]\rew_{\Gg^\dag}\Nn$ only for $\Nn=\set{P,\tr(\emptyset,Z,L)}$.
		We know that $\tr(\emptyset,Z,E)$ is $\Gg^\dag$-convergent in $n$ steps, so, by the above, $P$ is $\Gg^\dag$-convergent in less than $n$ steps.
		In consequence $\Ff=\set{z'}$ satisfies the thesis, because $E\erew_\Gg\Ff$, by definition.
	\item	Finally, suppose that $E=\esubst{E_0}{L}{z}$, where $E_0$ is not a variable.
		This case was completely resolved in \cref{sec:correctness}, so we do not repeat the proof here.
	\end{bracketenumerate}

\end{document}